\documentclass[dvipdfmx, 12pt]{article}

\usepackage[dvipdfmx]{graphicx,color}
\usepackage{bm}
\usepackage{amsthm}
\usepackage{mathrsfs}
\usepackage{amsmath}
\usepackage{amssymb}
\usepackage{amsfonts}
\usepackage{natbib}
\usepackage{algorithmic}
\usepackage{algorithm}
\usepackage{url}
\usepackage{multirow}

\usepackage{newtxtext,newtxmath}
\usepackage[top=30truemm,bottom=30truemm,left=25truemm,right=25truemm]{geometry}
\usepackage{here}
\usepackage{enumerate}
\usepackage{xcolor}

\newcommand{\field}[1]{\mathbb{#1}}

\newcommand{\N}{\mathbb{N}}

\newcommand{\dif}{\,\mathrm{d}}

\newcommand\mi{\mathrm{i}}

\newcommand{\R}{\field{R}}
\newcommand{\T}{\top}

\newcommand{\Z}{\field{Z}}

\newcommand{\abs}[1]{\lvert#1\rvert}

\newcommand{\inp}[2]{\langle#1,#2\rangle}
\newcommand{\norm}[1]{\lVert#1 \rVert}

\newcommand{\bX}{\bm{X}}
\newcommand{\bY}{\bm{Y}}

\newcommand{\bmu}{\bm{\mu}}
\newcommand{\bomega}{\bm{\omega}}
\newcommand{\bpsi}{\bm{\bm{\psi}}}

\newcommand{\uint}{\int^{1}_{0}}

\newtheorem{thm}{Theorem}[section]

\newtheorem{lem}[thm]{Lemma}

\theoremstyle{definition}

\newtheorem{asp}[thm]{Assumption}

\newtheorem{rem}[thm]{Remark}

\makeatletter

    \@addtoreset{equation}{section}
\makeatother

\title{An information criterion for detecting periodicities in functional time series} 
\date{June 6, 2026} 
\author{
Rinka Sagawa\\
{\small Department of Applied Mathematics}\\
{\small Waseda University}\\
{\small169-8555, Tokyo, Japan}
\and
Yan Liu\\
{\small Faculty of Science and Engineering}\\
{\small Waseda University}\\
{\small 169-8555, Tokyo, Japan}
\and
Valentin Patilea\\
{\small Centre de Recherche en \'{E}conomie et Statistique}\\
{\small  \'{E}cole Nationale de la Statistique et de l'Analyse de l'Information}\\
{\small Campus de Ker-Lann, rue Blaise Pascal, 35172 Bruz cedex, France}
}

\begin{document}
\maketitle

\begin{abstract}
We propose an information criterion for determining an unknown number of periodic components in functional time series. Identifying the number of frequencies in large-scale time series has been a central focus. To achieve this goal, we suggest an iterative procedure, utilizing the residual process obtained through least squares fitting. This iterative approach demonstrates broad applicability. We establish the consistency of the estimated number of periodic components 
by minimizing the information criterion. 
The efficacy of the procedure is illustrated through numerical simulations. 
In real data analysis, we apply this information criterion to temperature data and sunspot data.
\end{abstract}

\section{Introduction}

Functional data analysis has been a focal topic for enhancing predictive performance in complex data analysis.
This topic has been covered in several monographs; for example, \cite{bosq2000linear}, \cite{Ramsay2002} and \cite{kokoszka2017introduction},
just to name a few.
The practical application of functional data analysis spans a wide range of disciplines, including criminology, economics, archaeology, rheumatology, psychology, neurophysiology, auxology, meteorology and biomechanics.

Functional time series consist of 
functional observations indexed in time order.
More formally, a functional time series is a sequence of random functions 
$\{Y_t(u); u\in [0,1], t\in \mathbb{Z}\}$, where each $Y_t$ is a random element in $L^2([0,1])$.
For example, this type of data can be obtained by segmenting original data 
into smaller intervals.
Let us consider the daily average temperature data of Kyoto in Japan for $3$ years (See Figure~\ref{fig_fts}).
By dividing the original data into 3 smaller intervals of $365$ days each,
it reveals that there exist specific cycles in changes of the daily average temperature.
Even within the same dataset,
different lengths of intervals
may result in different observations of functional time series.

The statistical inference for functional time series
has been extensively explored so far (\cite{Hoermann2010, Hoermann2012}).
The prediction problem for functional autoregressive processes
has been considered by multivariate statistical techniques in \cite{aue2015prediction}.
Testing for periodicity using the asymptotic null distribution of the functional ANOVA statistics has been established by
\cite{hormann2018testing}.
The investigation has also been extended to the frequency domain, 
leading to the construction of spectral density operators for functional time series in a separable Hilbert space,
with applications to test for second-order stationarity 
(e.g., \cite{Delft2020}, \cite{van2020note}, \cite{aue2020testing}).

\begin{figure}[h]
\centering
\includegraphics[width=11cm, height=0.8\textheight, keepaspectratio]{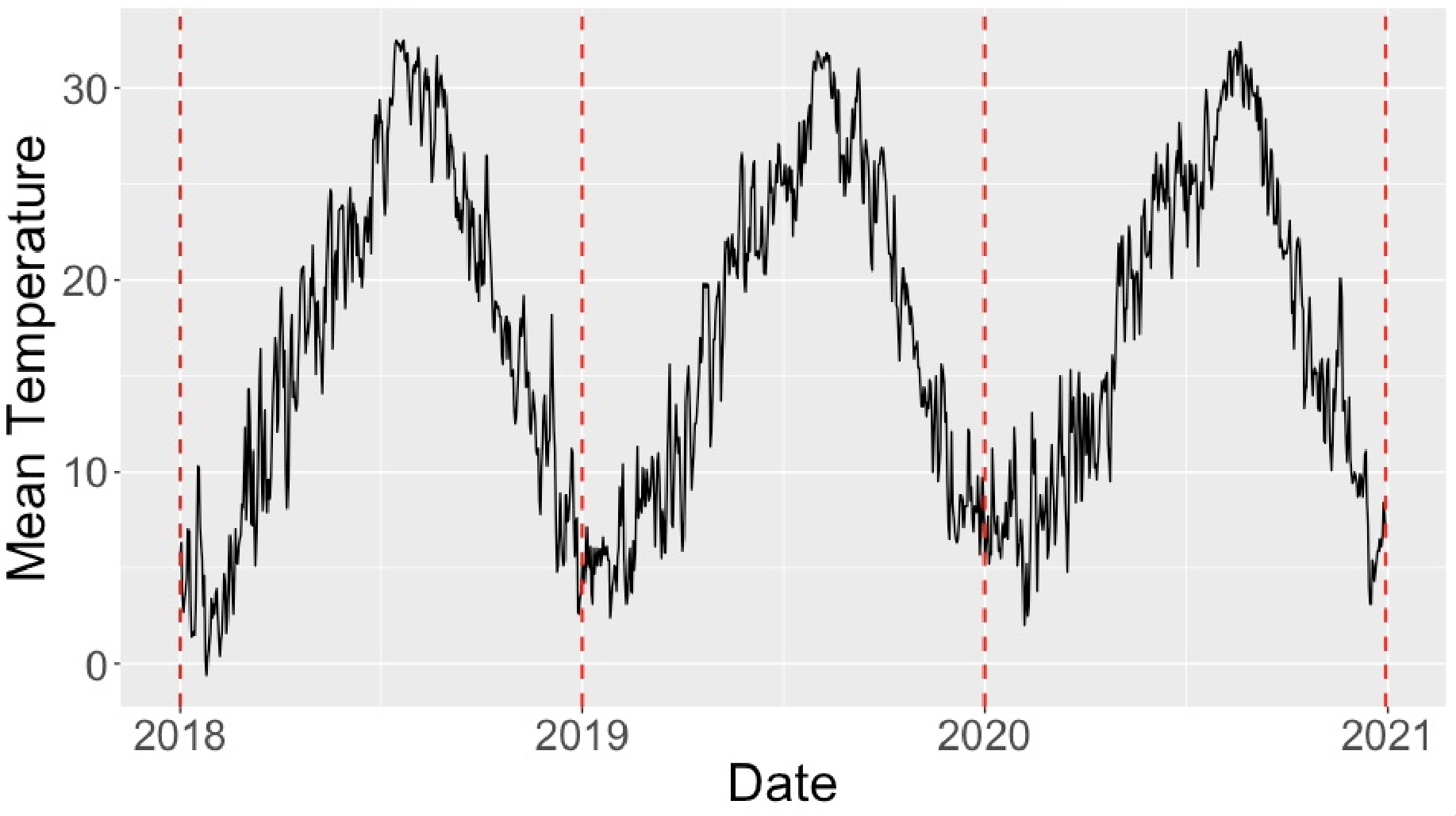}
\caption{
The daily average temperature data of Kyoto in Japan, from January 1, 2018 to December 25, 2020.
The dashed red lines indicate the segmentation of the data into 3 intervals of $365$ days each.
}
\label{fig_fts}
\end{figure}

In this paper, 
we consider a model of functional time series with trigonometric regression components.
An explicit expression of trigonometric functions in the model
provides a clear and interpretable representation of periodic structures, enabling consistent estimation of  periodic components.
Under this setting, we propose an information criterion for detecting periodicities in this model.
Determining the number of periodic components by a simple information criterion
distinguishes our approach from existing literature.
Our method employs a BIC-type model selection criterion,
which has been introduced by \cite{akaike1977entropy} and \cite{schwarz1978estimating}.
We suggest applying the information criterion to the empirical functional principal components
of the functional time series.
It is shown that the parameter of the trigonometric regressors within the functional time series 
converges in probability to the true parameter. 
We also establish the consistency of determining the number of periodicities by minimizing the information criterion.
The numerical simulations illustrate that
the selection of the true model
is not sensitive to the choice of the penalty factor 
included in the proposed criterion.
It should be remarked that this approach is different from the order selection for autoregressive models.
This new selection criterion is applied to temperature data and sunspot data in our real data analyses.

A comprehensive discussion on time series analysis has been structured in \cite{Brockwell1991}, \cite{Taniguchi2000} and \cite{Shumway2000}.
In a general framework, the regression model
for time series analysis has been thoroughly considered.
The statistical inference for multiple periodicities
was considered in \cite{hannan1973estimation}.
The model selection techniques for a single time series were considered in
\cite{quinn1989estimating}, \cite{wang1993aic} and \cite{kavalieris1994determining}.
The analysis of sunspot data by time series method
was considered in \cite{Kunsch1989}.
A nonparametric estimation method was proposed in \cite{Vogt2014}
to find out the anomalies in yearly global temperature.
\cite{patilea2016testing} considered the goodness-of-fit for a regression model with a functional response.
\cite{proietti2023seasonality} considered high-frequency time series
to model the seasonality in time series data.
Based on the previous studies,
we propose a BIC-type information criterion to determine the number of periodicities in functional time series.
Our proposal works well in real data analyses for temperature data and sunspot data.

The contributions of this paper can be summarized in the following three points.
First, the number of periodicities in functional time series 
can be automatically determined by our proposed information criterion,
which avoids the multiple testing issue.
Second, the consistency of our proposed procedure was shown theoretically,
which guarantees the detection of multiple frequencies in functional time series.
Finally, unlike the above literature to consider an information criterion for a single time series,
the procedure can be regarded as one for multiple time series obtained from functional time series.
In other words, this allows us to determine the number of periodicities for both the multivariate time series and the functional time series.

The remainder of the paper is organized as follows. 
In Section $\ref{sec2}$, we describe the parameter estimation of the regressors
and periodicities in the regression model for functional time series.
The estimated parameter vector is shown to converge 
to the true one in probability. 
In Section $\ref{sec3}$, an information criterion for detecting the number of periodicities is proposed. 
We express the procedure to determine the number of periodicities based on the information criterion in an algorithmic way.
In Section $\ref{sec4}$, 
numerical simulations reveal that
the performance of our selection procedure
for the number of periodicities
is insensitive to 
the choice of the penalty factor 
included in the proposed criterion.
In Section \ref{sec5}, 
we conduct real data analyses to determine the number of periodicities in both temperature data and sunspot data.
Section \ref{sec:6} concludes the paper.
The proof of theorem is presented in Appendix.
The proofs of technical results
and complete results of simulations and data analysis are relegated to the Supplementary Material.

\section{Trigonometric regression models}
\label{sec2}
In this section, 
let us consider the trigonometric regression model of order $r_{0}$ ($r_{0}$ is provisionally known) with functional time series.
Suppose $\{X_t; \, t\in \Z\}$
is a zero-mean stationary time series of functions in $\mathcal{H}:= L^2([0,1])$,  
which is a space of square integrable functions $g:[0,1] \to \R$, equipped with the inner product
\[
\inp{g_1}{g_2} = \uint g_1(u) g_2(u) \dif u, \qquad g_1, g_2 \in \mathcal{H},
\]
and the corresponding norm $\norm{\, \cdot \,}_{\mathcal{H}}$.
By definition, the covariance operator $\Gamma_0$ of the process is
\[
\Gamma_0(\,\cdot\,) = \mathbb{E}[\inp{X_t}{\, \cdot \, } X_t].
\]
In addition, we introduce the cross-covariance operator between $X_0$ and $X_t$ as
\[
\Gamma_{t}(\,\cdot\,) = \mathbb{E}[\inp{X_0}{\, \cdot \, } X_t],
\]
which coincides with $\Gamma_0$ when $t = 0$.

A trigonometric regression model with functional time series is
\begin{equation}\label{FTS_model}
Y_t(u) = \mu(u) +\left(\sum_{k=1}^{r_{0}}
\bigl(\alpha_k \cos(t\theta_k)+\beta_k \sin(t\theta_k) \bigr)\right)\omega (u)+ X_t(u), \qquad u \in[0,1],
\end{equation}
where $\mu$ and $\omega$ are unknown functions in $\mathcal{H}$
with
$\int_0^1 \omega^2(u)du=\norm{\omega(u)}_{{\mathcal{H}}}^2=1$.
For each $i = 1, \ldots, r_0$,
the parameters
$\alpha_{i}$, 
$\beta_{i}$ ($\alpha_i \not =0$ or $\beta_i \not = 0$), 
$\theta_{i}$ ($\in (0, \pi)$) are unknown;
$\theta_i \not = \theta_j$ if $i \not = j$;
if $r > r_0$, then $\alpha_r = \beta_r = 0$. 
The trigonometric regression model \eqref{FTS_model}
is the version of (2.5) of the model (M.2) in \cite{hormann2018testing}.


Let $(\nu_{\ell};\, \ell \in \N)$ be the orthonormal basis for $\mathcal{H}$
obtained through the functional principal component analysis (FPCA). 
With this FPCA basis, each $X_t$ can be represented using the Karhunen-Lo{\'e}ve representation
%
\[
X_t = \sum_{\ell = 1}^{\infty} \inp{X_t}{\nu_{\ell}} \nu_{\ell}.
\]
For a fixed constant $0 < p \in \N$, 
the functional principal component scores are
\begin{align*}
\bY_t &\equiv(\inp{Y_t}{\nu_1}, \inp{Y_t}{\nu_2}, \cdots, \inp{Y_t}{\nu_p})^{\T},\\
\bmu &\equiv(\inp{\mu}{\nu_1}, \inp{\mu}{\nu_2}, \cdots, \inp{\mu}{\nu_p})^{\T},\\
\bomega &\equiv (\inp{\omega}{\nu_1},\inp{\omega}{\nu_2}, \cdots, \inp{\omega}{\nu_p})^{\T},\\
\text{and} \qquad
\bX_t &\equiv (\inp{X_t}{\nu_1}, \inp{X_t}{\nu_2}, \cdots, \inp{X_t}{\nu_p})^{\T}.
\end{align*}
We impose the following assumptions for the identifiability of the model \eqref{FTS_model}.

\begin{asp}[Identifiability] \label{asp:iden}
\begin{enumerate}[(i)]
\item The parameters $\alpha_k$, $\beta_k$, $k = 1, \ldots, r_0$, and the function $\omega(u)$ are independent of $t$.
\item For the function $\omega(u)$, $\inp{\omega}{\nu_j}\neq 0$ for some $j \in \{1,\ldots, p\}$.
\end{enumerate}
\end{asp}
As a direct implication of Assumption \ref{asp:iden},
the vector $\bm{\omega}$ is non-zero, revealing the identifiability of the periodic components.
Assumption \ref{asp:iden} also indicates a 
guideline for deciding the dimension parameter $p$ in practice.
Theoretically, the larger the dimension $p$ is, the better 
the approximation performance of the function is.
Based on the view of detecting the number of periodicities,
however, 
the value of $p$ could be moderate when
Assumption \ref{asp:iden}(ii) is satisfied with
large coefficient parameters $\alpha_j$ and $\beta_j$.

This leads us to rewrite the model  \eqref{FTS_model} as
\begin{equation*}
\label{new_FTS_model}
\bY_t =\bmu + \left(\sum_{k=1}^{r_{0}}(\alpha_k \cos (t\theta_k)+\beta_k \sin (t\theta_k))\right) \bomega+\bX_t.
\end{equation*}
By construction, the cross-covariance matrix of $\{\bX_t\}$ is the $p \times p$-matrix 
\[
\Xi_{t}\equiv \bigl( \inp{\Gamma_{ t}(\nu_i)}{\nu_j}\bigr)_{i, j = 1, \ldots, p}.
\]

\begin{rem}
In practice, the orthonormal basis $(\nu_{\ell};\, \ell \in \N)$ 
is unknown in advance.
A practical approach is to use the empirical functional principal components instead of the true ones.
The basis is usually obtained from the observed stretch as follows.

Suppose now that we have observed $X_1, \ldots, X_N \in \mathcal{H}$.
The functional mean $\hat{\mu}_X$ is
$\hat{\mu}_X=\frac{1}{N}\sum_{t=1}^NX_t$, 
and the covariance operator is 
\[
\hat{\Gamma}_0(\cdot)=\frac{1}{N}\sum_{t=1}^N \inp{X_t-\hat{\mu}_X}{\cdot}(X_t-\hat{\mu}_X).
\]
 \cite{Hoermann2010} proved that  these estimators have $\sqrt{N}$-consistency 
 under the weak dependence assumption (e.g. $L^4$--$m$--approximability). 
 Here, $L^4$-$m$-approximability means that 
 the process $\{X_t\}$ can be approximated by a sequence $\{X_t^{(m)}\}$
 obtained by replacing innovations beyond lag $m$ with independent copies,
 such that the approximation error $\lVert X_t - X_t^{(m)}\rVert_4$
 decays sufficiently fast as $m \to \infty$.
From $\hat{\Gamma}_0(\cdot)$, for an arbitrary fixed but typically small $p<N$, 
the estimated eigenfunctions $\hat{\nu}_1,\ldots, \hat{\nu}_p$ can be computed,
and correspondingly, 
$\inp{Y_t}{\hat{\nu}_{\ell}}$, $\ell=1,\ldots, p$, are the empirical functional principal component scores.
\end{rem}

\begin{rem}
Our approach can also be considered with other bases of functions.
Let $\{\nu_i(u), i=1,\ldots, p\}$ be a class of basis functions, e.g., Fourier bases or B-spline bases.
Then the functional data $X_t(u)$ is approximated by
the following approximation:
\[
(\hat{\gamma}_1, \ldots, \hat{\gamma}_p)
= \arg \min_{\bm{\gamma}}\norm{X_t - \sum_{i=1}^p \gamma_i \nu_i}^2.
\]
See \cite{ramsay2006functional} for details.
\end{rem}

Let us use $\bm{\alpha}_k \equiv \alpha_k \bm{w}$ and $\bm{\beta}_k\equiv \beta_k \bm{w}$ for $k = 1, \ldots, r_{0}$,
where clearly $\bm{\alpha}_k \not = \bm{0}$ and $\bm{\beta}_k \not = \bm{0}$ for each $k$.
We arrive at the following trigonometric regression model
\begin{equation}\label{new_FTS_model}
\bY_t =\bmu + \sum_{k=1}^{r_{0}}(\cos (t\theta_k) \bm{\alpha}_k + \sin (t\theta_k) \bm{\beta}_k) +\bX_t.
\end{equation}
To keep the brevity, let $\bm{\psi}(r)$ be the vector of unknown parameters, i.e.,
\[
\bm{\psi}(r)\equiv
(
\bmu^{\T}, 
\bm{\alpha}_1^{\T}, \bm{\beta}_1^{\T}, \bm{\alpha}_2^{\T}, \bm{\beta}_2^{\T}, \ldots, \bm{\alpha}_{r}^{\T}, \bm{\beta}_{r}^{\T})^{\T}
\in \mathbb{R}^{(2r + 1 )p\times 1},\qquad r=0,\cdots,r_0.
\]
Accordingly, let $\bm{q}_t(r)$ be the vector of trigonometric functions, i.e.,
\[
\bm{q}_t(r)  = 
\bigl(1, \cos (t\theta_1), \sin (t\theta_1), \ldots, \cos (t\theta_{r}), \sin (t\theta_{r})
\bigr)^{\T},\qquad r=0,\cdots,r_0.
\]
Without any confusion, let $\bm{\psi}  = \bm{\psi}(r_0)$ and $\bm{q}_t  = \bm{q}_t(r_0)$.
The model \eqref{new_FTS_model} is 
now simplified in the following vector form:
\begin{equation}\label{new_FTS_model2_sec}
\bY_t = \bm{Q}_t(r_0) \bpsi + \bX_t,
\end{equation}
where 
$\bm{Q}_t(r)=(\bm{q}_t(r)^{\T} \otimes \bm{E}_p)\in \mathbb{R}^{p\times (2r + 1)p}$, $r=1,\ldots, r_0$, and $\bm{E}_p$ is the $p$-dimensional identity matrix.

Denote now the observed stretch of empirical functional principal components \eqref{new_FTS_model2_sec} by
$\bY_1, \ldots, \bY_N$.
Let $\bY$, $\bX$, and $\bm{Q}(r)$ be the matrices
$\bY=(\bY_1^{\T}, \ldots, \bY_N^{\T})^{\T}$, 
$\bX=(\bX_1^{\T}, \ldots, \bX_N^{\T})^{\T}$, 
and $\bm{Q}(r) = \bigl(\bm{Q}_1(r)^{\T},\ldots,\bm{Q}_N(r)^{\T} \bigr)^{\T} $
$\in \mathbb{R}^{Np \times (2r + 1)p}$,
$r=1,\ldots,r_0$,
respectively.
With this notation, the equation \eqref{new_FTS_model2_sec} can be rewritten under the vector for
\begin{equation} \label{eq:vector_mod_sec}
\bY = \bm{Q}(r_0) \bm{\psi} +\bX.
\end{equation}

Let $\hat{\bpsi}(r)$ be the least squares estimates of $\bpsi$ in \eqref{eq:vector_mod_sec} as
\begin{equation*} \label{eq:lse_psi}
\hat{\bpsi}(r) = 
\bigl(\bm{Q}(r)^{\T} 
\bm{Q}(r)\bigr)^{-1} \bm{Q}(r)^{\T} \bY,
\end{equation*}
and let $\hat{\bpsi} = \hat{\bpsi} (r_0)$.
Let $\Sigma\in \mathbb{R}^{Np\times Np}$ be the covariance matrix of $\bX$.
Then we have
\begin{equation*} \label{eq:noise_var}
\Sigma =
\begin{pmatrix}
\Xi_0 & \Xi_1 & \cdots & \Xi_{N - 1} \\
\Xi_1 & \Xi_0 & \cdots & \Xi_{N - 2} \\
\vdots & \vdots & \ddots & \vdots \\
\Xi_{N - 1} & \Xi_{N - 2} & \cdots & \Xi_0
\end{pmatrix}.
\end{equation*}
We impose the following dependence assumption
for the functional time series $X_t(u)$.
\begin{asp}[Cumulant kernel of order $k$]
\label{k_cum_ass}
Let $\mathrm{cum}_{t_1,\ldots,t_{k-1}}:\mathcal{H}^k\to \mathbb{R}$ be 
\[
\mathrm{cum}_{t_1,\ldots,t_{k-1}}(g_1,\ldots,g_k)
=\mathrm{cum}(\inp{X_0}{g_1}, \inp{X_{t_1}}{g_2},\ldots, \inp{X_{t_{k-1}}}{g_k})
\]
for $g_1,\ldots,g_k\in \mathcal{H}$. 
Here, $\rm{cum}$ denotes the joint cumulant of the random variables involved.
The series $\sum_{k=1}^{\infty}\mathcal{C}_kz^k/k!$ is convergent 
for $z$ in a neighborhood of $0$, where $\mathcal{C}_k$ is defined as
\[
\mathcal{C}_k
:=\sup_{g_1,\ldots, g_k \in \{\nu_1, \dots, \nu_p\} }\sum_{t_1,\ldots, t_{k - 1} } 
\abs{\mathrm{cum}_{t_1,\cdots,t_{k-1}}(g_1,\ldots, g_k)}.
\]
\end{asp}

Assumption \ref{k_cum_ass} is an extension of a dependence condition for multivariate time series,
which has been considered in \citet[Assumption 2.6.3]{brillinger2001time}.
It controls the dependence structure of the functional time series through the summability of higher-order cumulant kernels. 
This condition allows us to obtain almost sure bounds for various statistics of interest, 
which play a key role in establishing the consistency of the proposed information criterion.

\begin{rem}
Assumption \ref{k_cum_ass} characterizes the temporal dependence structure through higher-order cumulant kernels, and enables sharp theoretical results. In contrast, a commonly used notion of weak dependence in functional time series is
$L^{p}$-$m$-approximability, which is based on 
approximating the process by sequences obtained by replacing remote innovations with independent copies, so that the approximation error decays.
These two types of conditions are not directly comparable in general, as they capture different aspects of dependence. 
Under suitable regularity conditions, 
weak dependence structures ensuring good approximation properties are often associated with sufficiently fast decay of higher-order cumulants.
\end{rem}

\begin{rem}
In Assumption \ref{k_cum_ass}, the supremum is taken over the orthonormal basis $\{v_1, \ldots, v_p\}$, 
reflecting the functional time series data analysis in practice.
The truncation level $p$ determines the accuracy of this approximation and is therefore part of the modeling choice.
In special cases such as Gaussian functional time series, 
higher-order cumulants vanish, 
and the condition is automatically satisfied beyond second order. 
For more general processes, the choice of $p$
affects the behavior of higher-order cumulants through the finite-dimensional representation.\\
In practice, there is a trade-off between statistical accuracy and computational cost. 
Larger values of $p$ 
typically lead to more accurate representations, as also illustrated in our numerical results, while smaller values of $p$
result in lower computational burden.
\end{rem}

The consistency of the least squares estimates $\hat{\bpsi}(r)$ 
is shown in the following lemma.
\begin{lem}
\label{para_as_converge}
Suppose $\{X_t; \, t \in \Z\}$ is a zero-mean stationary process satisfying Assumption \ref{k_cum_ass}. 
Under Assumption \ref{asp:iden},
if $0 \leq r \leq r_0$,
then the least squares estimates $\hat{\bpsi}(r)$
converges to the true vector $\bpsi(r)$ in probability;
if $r > r_0$,
then the $k$th element of $\hat{\bpsi}(r)$, $k > (2r_0 + 1)p$, 
converges to 0 in probability,
as $N \to \infty$.
Especially, 
$\hat{\bpsi}$ converges to $\bpsi$ in probability.
\end{lem}


\begin{rem}
\label{rem:rem2}
\begin{enumerate}[(i)]
\item
Denote the $i$th element of $\bm{X}_t$ by $X^{(i)}_t$.
The only condition required to guarantee the consistency 
is the absolutely summable autocovariance sequence of each element $X^{(i)}_t$, $i = 1, \dots, p$.
It has been shown in Lemma 4.1 in \cite{Hoermann2010}
that the $L^2$--$m$-approximable sequence has this property.
\item
The mean vector $\bm{\mu}$ can be estimated consistently.
\cite{Hoermann2010} has already shown that $\mathbb{E}[\norm{\bar{\bY} - \bm{\mu}}^2] = O(N^{-1})$,
where $\bar{\bY}= N^{-1}\sum_{t=1}^N \bY_t$.
In addition, 
according to Lemma \ref{para_as_converge},
if the model order $r = r_0$ is specified,
then the estimator $\hat{\bpsi}$ converges to the true  parameter $\bpsi$ of trigonometric functions.
\end{enumerate}
\end{rem}

Let us move to the estimation problem of the periodicity parameter $\theta_k$, $k = 1, \ldots, r_0$.
To keep the brevity of the presentation, 
let $Y_{t}^{(j)}$ denote the $j$th element of the observed stretch $\bY_t$. 
Let $p_{\bY, N}$ be
\begin{equation*}
p_{\bY, N}(\theta)=\sum_{j=1}^p\left|N^{-1}\sum_{t=1}^N(Y_{t}^{(j)}-\bar{Y}^{(j)})e^{\mi t\theta}\right|^2.
\end{equation*}
In addition, 
let $\alpha_{kj}$ and $\beta_{kj}$ be the $j$th element of vectors $\bm{\alpha}_k$ and $\bm{\beta}_k$,
respectively.
It is easy to see that $\alpha_{kj}\cos (t\theta_k)+\beta_{kj}\sin (t\theta_k)$ 
represents the $j$th element of the second term in \eqref{new_FTS_model}.

\begin{lem} \label{p_y_approx}
Suppose $\{X_t; \, t \in \Z\}$ is a zero-mean stationary process satisfying Assumption \ref{k_cum_ass}. 
Under Assumption \ref{asp:iden}, it holds that
$p_{\bY,N}(\theta)=\sum_{j=1}^p|T_{j, N}(\theta)|^2+o_p(1)$ uniformly in $\theta \in (0, \pi)$,
where
\[
T_{j, N}(\theta)=N^{-1}\sum_{t=1}^N\left\{\sum_{k=1}^{r_{0}}[\alpha_{kj}\cos (t\theta_k)+\beta_{kj}\sin (t\theta_k)]e^{\mi t\theta}\right\}.
\]
Additionally, as $N \to \infty$, it holds that
\[
\sum_{j=1}^p\abs{T_{j, N}(\theta)}^2\to
\begin{cases}
0, & \text{if $\theta\not \in \{\theta_1, \dots, \theta_{r_0}\}$}, \\
\sum_{j=1}^p(\alpha_{rj}^2+\beta_{rj}^2)/4, &
\text{if $\theta = \theta_r \in \{\theta_1, \dots, \theta_{r_0}\}$}.
\end{cases}
\]
Thus, we obtain
\[
p_{\bY, N}(\theta) \to_p 
\begin{cases}
0, & \text{if $\theta\not \in \{\theta_1, \dots, \theta_{r_0}\}$}, \\
\sum_{j=1}^p(\alpha_{rj}^2+\beta_{rj}^2)/4, &
\text{if $\theta = \theta_r \in \{\theta_1, \dots, \theta_{r_0}\}$}.
\end{cases}
\]
\end{lem}

According to Lemma \ref{p_y_approx}, 
if the periodic parameter $\theta$ appears in the model, then
the statistic $p_{\bm{Y}, N}(\theta)$ converges in probability to 
a linear combination of squared coefficient elements in the vectors $\bm{\alpha}_r$ and $\bm{\beta}_r$.
As a result, the periodicity can be estimated through the asymptotic limit of $p_{\bm{Y}, N}(\theta)$.
The true periodicity parameter $\theta_k$, $k = 1, \dots, r_0$, 
 can be recursively estimated by the following procedure.
Let $\bm{\psi}_{r_k}$, $\bm{q}_{t, r_k}$ and $\bm{Q}_{t, r_k}$ be
\begin{align*}
 \bm{\psi}_{r_k} &=
(
\bmu^{\T}, 
\bm{\alpha}_{r_1}^{\T}, \bm{\beta}_{r_1}^{\T}, \bm{\alpha}_{r_2}^{\T}, \bm{\beta}_{r_2}^{\T}, \ldots, \bm{\alpha}_{r_k}^{\T}, \bm{\beta}_{r_k}^{\T})^{\T}
\in \mathbb{R}^{(2k + 1 )p\times 1}, \\
\bm{q}_{t, r_k}  &= 
\bigl(1, \cos (t\theta_{r_1}), \sin (t\theta_{r_1}), \ldots, 
\cos (t\theta_{r_k}), \sin (t\theta_{r_k})
\bigr)^{\T}, \\
\text{and} \qquad
\bm{Q}_{t, r_k} &= (\bm{q}_{t, r_k}^{\T} \otimes \bm{E}_p) 
\in \mathbb{R}^{p\times (2k + 1)p}.
\end{align*}
The estimates of $\theta_{r_k}$ is obtained by
\begin{align}
\label{theta_estimate_eq}
\hat{\theta}_{r_k}&=
\arg \max_{\theta \in \Theta} 
p_{\bY,N}^{(k-1)}(\theta)\\ \notag
&=\arg \max_{\theta \in \Theta} \sum_{j=1}^p\left|N^{-1}\sum_{t=1}^N
\left(Y_{t}^{(j)}-
\sum_{l=1}^{k-1}\Bigl(\bm{Q}_{t, r_l} \hat{\bpsi}_{r_l} \Bigr)^{(j)}
\right)e^{\mi t\theta}\right|^2,
\quad k = 1, 2, \ldots, 
\end{align}
where 
$p_{\bY,N}^{(0)}(\theta): = p_{\bY, N}(\theta)$ and
$\Theta:= \{2\pi j/N;\, j = 1, \ldots, \lceil N/2 \rceil\}$.
The consistency of the estimator $\hat{\theta}_{r_k}$
is established in the following lemma.

\begin{lem}
\label{theta_estimate}
Suppose $r_k \in \{1, 2, \ldots, r_0\}$.
With the same assumption in Lemma \ref{p_y_approx},
it holds that 
$N(\hat{\theta}_{r_k} - \theta_{r_k})$ converges to $0$ in probability.
\end{lem}
The consistency of  $\hat{\bpsi}_{r_k}$ follows from Lemma \ref{para_as_converge}.
If $r_0$ is provisionally known,
then the periodicities $\{\theta_1, \ldots, \theta_{r_0}\}$ are correctly
specified by $\{\theta_{r_1}, \ldots, \theta_{r_{r_0}}\}$ in asymptotics.
Even if some $\theta_{r_l} \not \in \{\theta_1, \ldots, \theta_{r_0}\}$,
by the consistency shown in Lemma \ref{para_as_converge},
$\hat{\bpsi}_{r_l}$ converges to $0$ in probability.
Therefore, the procedure is robust to the overspecification
of the number of periodicities.

In summary, 
the parameter estimation problem for functional harmonic regression models
has been considered in this section.
The consistency for the estimation of each parameter
has been established, 
as functional time series are reduced to multivariate ones.
Hereafter, we develop an information criterion 
following our parameter estimation,
which has not been considered in the literature for multivariate time series.

\section{Sample-based selection of the number of periodicities}
\label{sec3}
In this section, we propose an information criterion for 
specifying the unknown parameter $r_{0}$.
Let us recall that the cross-covariance matrix is denoted by $\Xi_{t}$.
The new information criterion is based on the first principal component
of the matrix $\Xi_0$.
In other words, let $\nu_1$ be
\begin{equation} \label{eq:pca}
\nu_1 = \arg \max_{\nu \not= \bm{0}} \frac{\nu^{\T} \Xi_0 \nu}{
\nu^{\T} \nu}.
\end{equation}
The linear combination $\nu_1^{\T} \bm{X}_t (:=\tilde{X}_t$, say) 
is known as the first principal component (e.g., \citet[p.438]{Shumway2000}).
We remark that,
although the basis functions in the Karhunen-Lo{\'e}ve expansion 
are theoretically determined by the covariance operator,
the specific choice of basis is not critical in practice, 
when using the approach below based on the first principal component \eqref{eq:pca}. 

The adoption of the first principal component
keeps consistency with the approach of using the empirical functional principal component
in Section \ref{sec2}.
Even if the orthonormal basis of $\mathcal{H}$
is chosen only as a set of linearly independent functions,
our proposed information criterion still works
under the consideration based on the first principal component.
Furthermore,
$\nu_1^{\T} \bm{Y}_t (:=\tilde{Y}_t$, say) 
is a linear transformation of $\bm{Y}_t$,
so that 
the frequency parameters $\theta_k$, $k = 1, \ldots, r_0$,
are invariant under this linear transformation.

\begin{rem}
Although the frequency parameters $\theta_k$, $k = 1, \ldots, r_0$, are invariant under the linear transformation, 
the practical detectability of periodic components depends on 
the strength of their projection onto the first principal component direction. 
If the periodic signal is nearly orthogonal to $\nu_1$, 
the resulting signal-to-noise ratio may be reduced, potentially affecting the detection performance.
\end{rem}

From \eqref{new_FTS_model2_sec}, the model now is
\begin{equation}\label{linear_transformation_model}
\tilde{Y}_t=\bm{q}_t(r_0)^{\T}\tilde{\bpsi}+\tilde{X}_t,
\end{equation}  
where 
\begin{align*}
\tilde{\bpsi}&=(E_{(2r_0+1)}\otimes \nu_1^{\T})\bpsi\\
&=(\inp{\nu_1}{\bm{\mu}}, 
\inp{\nu_1}{\bm{\alpha}_1}, 
\inp{\nu_1}{\bm{\beta}_1},\ldots,
\inp{\nu_1}{\bm{\alpha}_{r_0}}, 
\inp{\nu_1}{\bm{\beta}_{r_0}})^{\T}\\
&=(\tilde{\mu},\tilde{\alpha}_{1}, \tilde{\beta}_{1},\ldots,\tilde{\alpha}_{r_0},\tilde{\beta}_{r_0})^{\T} \in \mathbb{R}^{(2r_0+1)\times 1}.
\end{align*}
The second identity follows $\mathrm{vec}(ABC)=(\bm{E} \otimes AB)\mathrm{vec}(C)$ (e.g., \citet[p.662]{lutkepohl2005new}).

Now, we consider the prediction error
by fitting an $h$-order autoregressive model  to the first principal component $\tilde{X}_t$,
and denote the error by $\hat{\sigma}^2(h)$.
Let $\hat{X}_t(r)$ be the residuals of 
linear regression on
all trigonometric functions $\cos (t\hat{\theta}_k)$ and $\sin (t\hat{\theta}_k), k=1,\ldots,r$, such as
\[
\hat{X}_t(r)=\tilde{Y}_t-\hat{\bm{q}}_t(r)^{\T}\hat{\tilde{\bpsi}}(r),
\]
where $\hat{\bm{q}}_t(r) = \bigl(1, \cos (t \hat{\theta}_1), \sin (t \hat{\theta}_1), \ldots, \cos (t \hat{\theta}_{r}), \sin (t \hat{\theta}_{r})
\bigr)^{\T}$, and $\hat{\tilde{\bpsi}}(r)$ is the least squares estimates of $\tilde{\bpsi}(r)$.
Accordingly, 
an approximate prediction error 
$\hat{\sigma}_r^2(h)$ is naturally defined in terms of 
the residuals $\hat{X}_t(r)$ by fitting an $h$-order autoregressive model to $\hat{X}_t(r)$.

We propose the following criterion to detect the true number $r_0$ of periodicities:
\begin{equation}\label{criterion}
\varphi(r, h)=\log \{\hat{\sigma}_r^2(h)\}+(\kappa r+h)\frac{\log N}{N},
\end{equation}
where $\kappa:=\kappa_N$ is some positive constant. 
To be specific,
for each number $r$ of periodicities,
we can find out an autoregressive model of order $h$
so that the model minimizes $\varphi(r, h)$;
let $\hat{h}_r$ be the minimizer.
We compare the different values of the criterion $\varphi(r, \hat{h}_r)$,
and choose $\hat{r}$
as the minimizer of $\varphi(r, \hat{h}_r)$ for estimating the true number $r_0$.

\begin{asp}
\label{h_condition}
Let $h:=h_N$ be a sequence such that
$h\to \infty$ and $h = O(\log N/\log \log N)$, as $N \to \infty$.\\
\end{asp}
A feasible and natural choice is $h = \log \log N$, which satisfies Assumption \ref{h_condition}.

To reveal the performance of our proposed information criterion,
we briefly explain the frequency domain framework for functional time series.
Let $\mathcal{F}_{\theta}: \mathcal{H} \to \mathcal{H}$ be the spectral density operator at frequency $\theta$ for 
the functional time series $X_t(u)$.
Assuming $\sum_{t\in \mathbb{Z}}\norm{\Gamma_t}_1<\infty$, where $\norm{\cdot}_1$ denotes the trace norm, 
and following \cite{Panaretos2013}, the spectral density operator $\mathcal{F}_{\theta} : \mathcal{H}\to \mathcal{H}$ has the expression
\[
\mathcal{F}_{\theta}(\cdot)
=
\frac{1}{2\pi}\sum_{t\in \mathbb{Z}}\exp (-\mi\theta t)\Gamma_t(\cdot).
\]
The spectral density function of $\tilde{X}_t$ is $f_{\tilde{X}}(\theta) = \inp{\mathcal{F}_{\theta}(\nu_1)}{\nu_1}$
in observing the orthonormal basis $(\nu_{\ell};\, \ell = 1, \ldots, p)$.
Theoretically, we obtain the following lemma for 
the prediction error $\hat{\sigma}^2(h)$ and the approximate error $\hat{\sigma}_r^2(h)$, $r = 0$.
\begin{lem}
\label{error_evaluate}
Suppose $\{X_t; \, t \in \Z\}$ is a zero-mean stationary process satisfying Assumption \ref{k_cum_ass}.
Under Assumptions \ref{asp:iden} and \ref{h_condition}, 
we obtain the following approximation
\[
\hat{\sigma}^2_0(h)=
\hat{\sigma}^2(h)+\sum_{k=1}^{r_{0}}\frac{|\eta_h(e^{\mi \theta_k})|^24\pi f_{\tilde{X}}(\theta_k)}{h}+o(h^{-1}),
\]
where $\eta_h(e^{\mi \theta})=1+\sum_{j=1}^h\eta_je^{\mi j\theta}$ with coefficients
$\eta_1, \ldots, \eta_h$ such that 
\begin{equation} \label{eq:best_linear}
\mathbb{E}\abs{\tilde{X}_t-\eta_1\tilde{X}_{t-1}-\ldots-\eta_h\tilde{X}_{t-h}}^2
=
\min_{b_1,\ldots,b_h} \mathbb{E}\abs{\tilde{X}_t-b_1\tilde{X}_{t-1}-\ldots-b_h\tilde{X}_{t-h}}^2.
\end{equation}
\end{lem}
Here, 
$\eta_h(z) = 1 + \sum_{j=1}^h \eta_j z^j$ is the inverse filter of the best linear predictor as in \eqref{eq:best_linear}.
From Lemma \ref{error_evaluate}, it is found that the difference between $\hat{\sigma}_0^2$ and $\hat{\sigma}^2$ 
has an expression of relevant quantities
$f_{\tilde{X}}(\theta_k)$, the spectral density function of $\tilde{X}_t$,
and the summation running from 1 to the true number $r_0$ of periodicities. 

We now complete the algorithm to estimate $r_0$ by the information criterion $\varphi(r, h)$ with an upper bound $H$.
From a theoretical perspective, $H$ is assumed to satisfy $H=o(N^{1/4})$ to ensure the consistency,
while in practice,
a fixed moderate value of $H$ typically performs well in finite samples.

\begin{algorithm}[H]
\caption{The algorithm of detecting the number $r_0$ of periodicities.}
\label{al1}
\begin{tabbing}
   \qquad \enspace Set : $r=0$.\\
   \qquad \qquad Step 1\enspace  For $h\leq H$, fit an $h$-order autoregressive model to $\hat{X}_t(0)$ to compute $\hat{\sigma}_0^2(h)$.\\
   \qquad \qquad Step 2\enspace   Minimize $\varphi(0,h)$ with respect to $h$ to obtain $\varphi(0,\hat{h}_0)$.\\
   \qquad \qquad Step 3\enspace   For fixed $r$, estimate the $(r+1)$th frequency $\hat{\theta}_{r+1}$ by utilizing \eqref{theta_estimate_eq}.\\
  \qquad \qquad Step 4\enspace   For $h\leq H$, fit an $h$-order autoregressive model to $\hat{X}_t(r+1)$ to compute $\hat{\sigma}_{r+1}^{2}$.\\
   \qquad \qquad Step 5\enspace Minimize $\varphi(r+1,h)$ with respect to $h$ to obtain $\varphi(r+1,\hat{h}_{r+1})$.\\
\qquad \qquad \qquad \enspace If $\varphi(r+1,\hat{h}_{r+1})<\varphi(r,\hat{h}_r)$,\\
\qquad \qquad \qquad \qquad \enspace Repeat Step 3 through Step 5 with $r \leftarrow r+1$.\\
\qquad \qquad \qquad \enspace Else  \\
\qquad \qquad  \qquad \qquad \enspace Stop the recursion and obtain $\hat{r}=r$.\\
\qquad \enspace Output : The estimated number $\hat{r}$ of periodicities.
\end{tabbing}
\end{algorithm}

\noindent

The following main result provides the theoretical justification for the Algorithm \ref{al1} 
by showing that the estimated number of periodicities in the output is consistent.
\begin{thm}
\label{criterion_thm}
Suppose $\{X_t; \, t \in \Z\}$ is a zero-mean stationary process satisfying Assumption \ref{k_cum_ass}.
Under Assumptions \ref{asp:iden} and \ref{h_condition}, $\hat{r}$ converges to $r_0$ in probability.
\end{thm}

\noindent
Therefore, the estimated number $\hat{r}$ of periodicities by utilizing Algorithm \ref{al1} has the consistency.

\section{Simulation}
\label{sec4}
In this section, we verify that the proposed criterion is insensitive to the choice of the parameter $\kappa$ as 
the length of observation $N$ increases through numerical simulations.
For the proposed criterion \eqref{criterion}, 
the main concern is that the value of $\kappa$ may potentially have influence on 
the model selection of true number of periodicities.
To alleviate this concern, we check the  ``stable'' range of  $\kappa$ 
by evaluating the number of simulations arriving at the true number of periodicities among all simulations.
The term ``stable'' refers to the ability to 
correctly estimate the number of periodicities most frequently
across all simulations.
We assess the ``stable'' range of $\kappa$ by $100$ simulations.
Additionally, we present the optimal range 
of the parameter $\kappa$ based on the following criterion.
The optimal range of $\kappa$ is determined
when 
the rate of correctly estimating the number of periodicities 
is equal to or greater than $90\%$.

First, we consider the following model with $r_0 = 3$:
\begin{equation}
\label{true_model}
Y_t(u)=\left(\cos \left(\frac{2\pi}{5}t\right)+\cos \left(\frac{2\pi}{6}t\right)+\cos \left(\frac{2\pi}{15}t\right)\right)(1+u^2)+X_t(u).
\end{equation}
Here,  $X_t(u)$ is a stationary functional AR ($2$) model:
\[
X_t(u)=\Phi_1(X_{t-1}(u))+\Phi_2(X_{t-2}(u))+\epsilon_t(u),\qquad u\in[0,1],
\]
where $\epsilon_t(u)$ is a sequence of  i.i.d.~standard Gaussian elements in $\mathcal{H}$.
In other words, 
all projections $\inp{\epsilon_t(u)}{\nu}$ for $\nu\in \mathcal{H}$ are normally distributed with mean $0$ and variance $\inp{\Gamma(\nu)}{\nu}=1$.
Let 
$\Phi_1:\mathcal{H}\to \mathcal{H}$ be the coefficient operator satisfying, for cubic \emph{B}-spline basis functions $\nu_1,\ldots,\nu_p$ in $\mathcal{H}$,
\[
\begin{cases}
\inp{\Phi_1(\nu_i)}{\nu_j}=0.2&(i=j=1,\ldots,p),\\
\inp{\Phi_1(\nu_i)}{\nu_i}=0&(i\neq j)
\end{cases}
\]
with $p=30$.
Additionally, $\Phi_2:\mathcal{H}\to \mathcal{H}$ is the coefficient operator satisfying, 
for orthonormal basis functions $\nu_1,\ldots,\nu_p$ in $\mathcal{H}$, and positive integers $s=1,\ldots,[(p+2)/3]$,
\[
\begin{cases}
\inp{\Phi_2(\nu_{3s-2})}{\nu_{3s-2}}=0.7,&\\
\inp{\Phi_2(\nu_{3s-1})}{\nu_{3s-1}}=-0.5,&\\
\inp{\Phi_2(\nu_{3s})}{\nu_{3s}}=\inp{\Phi_2(\nu_{3s})}{\nu_{3s-2}}=0.3,&\\
\inp{\Phi_2(\nu_{3s})}{\nu_{3s-1}}= -0.1,&\\
\inp{\Phi_2(\nu_{3s-1})}{\nu_{3s}}=\inp{\Phi_2(\nu_{3s-1})}{\nu_{3s-2}}=0,&\\
\inp{\Phi_2(\nu_{3s-2})}{\nu_{3s}}=\inp{\Phi_2(\nu_{3s-2})}{\nu_{3s-1}}=0.&
\end{cases}
\]
Note that $0 < p\in \mathbb{N}$ and $3s$, $3s-1$, $3s-2 \leq p$. 

\begin{table}[H]
\centering
\caption{Full results of the ``stable'' range of $\kappa$ when the 100 simulations are generated from the true model and each simulation has $960$ observations of functional time series. }
\label{tab:all_omega_960}
\begin{tabular}{ccccccccccc}
\hline
     & $\hat{r}=0$ & $\hat{r}=1$ & $\hat{r}=2$ & $\hat{r}=3$ & $\hat{r}=4$ & $\hat{r}=5$ & $\hat{r}=6$ & $\hat{r}=7$ & $\hat{r}=8$ & $\hat{r}=9$ \\ \hline
$\kappa=1$  &  0   &  6   & 0    & $\bm{51}$    & 15    & 5    &  4   & 4    &  15   & 0    \\
$\kappa=2$  & 0    &  6   &  0  & $\bm{81}$    &  11   & 2    & 0    &  0   & 0    & 0    \\
$\kappa=3$  &  0   & 6    & 0    & $\bm{92}$    &  2   &  0   & 0    & 0    & 0    & 0    \\
$4\leq \kappa \leq 27$  & 0    & 6    &  0   & $\bm{94}$    &  0   &  0   &  0   &  0   &  0   &  0   \\
  $\kappa=28$   &  1   &  6   &  0   & $\bm{93}$    &  0   &  0   &  0   &  0   &  0   & 0    \\
 $29\leq \kappa \leq 48$  &  2   &  6   &  0   & $\bm{92}$    &  0   &  0   &  0   &  0   &  0   & 0    \\
$\kappa=49$  &  2   &  6   &  1   & $\bm{91}$    &  0   &  0   &  0   &  0   &  0   & 0    \\
$50 \leq \kappa \leq 53$  &  2   &  6   &  2   & $\bm{90}$    &  0   &  0   &  0   &  0   &  0   & 0    \\
$54\leq \kappa \leq 55$  &  2   &  6   &  3   & $\bm{89}$    &  0   &  0   &  0   &  0   &  0   & 0    \\
$\kappa=56$  &  2   &  6   &  4   & $\bm{88}$    &  0   &  0   &  0   &  0   &  0   & 0    \\
$57\leq \kappa \leq 58$  &  2   &  6   &  5   & $\bm{87}$    &  0   &  0   &  0   &  0   &  0   & 0    \\
$\kappa=59$  &  3   &  6   &  10   & $\bm{81}$    &  0   &  0   &  0   &  0   &  0   & 0    \\
$\kappa=60$  &  3  &  6   &  13   & $\bm{78}$    &  0   &  0   &  0   &  0   &  0   & 0    \\
$\kappa=61$  &  4  &  6   &  17   & $\bm{73}$    &  0   &  0   &  0   &  0   &  0   & 0    \\
$62\leq \kappa \leq 63$  &  4  &  6   &  19   & $\bm{71}$    &  0   &  0   &  0   &  0   &  0   & 0    \\
$\kappa=64$  &  4  &  6   &  23   & $\bm{67}$    &  0   &  0   &  0   &  0   &  0   & 0    \\
$\kappa=65$  &  4  &  6   &  26   & $\bm{64}$    &  0   &  0   &  0   &  0   &  0   & 0    \\
$\kappa=66$  &  4 &  6   &  28   & $\bm{62}$    &  0   &  0   &  0   &  0   &  0   & 0    \\
$\kappa=67$  &  4 &  6   &  32   & $\bm{58}$    &  0   &  0   &  0   &  0   &  0   & 0    \\
$\kappa=68$  & 4  &  6   &  36   & $\bm{54}$    &  0   &  0   &  0   &  0   &  0   & 0    \\
$\kappa=69$  &  4  &  6   &  41   & $\bm{49}$    &  0   &  0   &  0   &  0   &  0   & 0    \\
   \hline
\end{tabular}
\end{table}

In the simulation, we take the maximum of $h$ as $H = 8$,
the range of $\kappa$ as $1 \leq \kappa \leq 69$, and $r\leq 9$.
The results for other values of $\kappa$
are omitted in case that $\kappa$ is out of the ``stable'' range.
The result for $N=960$ is  shown in Table \ref{tab:all_omega_960}.
The results for $N=120$ and $N=480$ are provided in the Supplementary Material.
The plots in Figure \ref{fig:performance_hatr3}
show the rate when the period is correctly estimated
across all simulations for each $\kappa$ within the ``stable" range when  $N=120$, $480$, and $960$.

From Table \ref{tab:all_omega_960} and Figure \ref{fig:performance_hatr3}, 
our proposed criterion \eqref{criterion} is insensitive to 
the choice of $\kappa$ when the sample size $N$ is sufficiently large. 
Especially, $\kappa$ ranging from 4 to 11 is optimal 
from the perspective of the higher rate 
of correctly estimating the true number of periodicities for 
all different lengths of observations.

Next, we fix the hyperparameter $\kappa$ in the proposed criterion as $\kappa=5$,
and take different basis functions $\nu_i$ with different number $p$ of basis functions into the numerical simulation for comparison.
 Specifically, we use three types of basis functions: B-spline basis, Fourier basis, and Wavelet (Haar wevelet) basis.
 The number of basis functions are $p=1$, 5, 10, 15, 20, 25, 30.
 Table \ref{tab:result_other_basis} summarizes the non-zero counts of the estimated number $\hat{r}$ by Algorithm \ref{al1} 
 based on 100 simulations for each basis type and number $p$ of basis functions.
 The length of observation is fixed as $N=960$ in each case.

In view of Table \ref{tab:result_other_basis}, the number of correct detections increases as $p$ increases, 
which aligns well with the theoretical understandings that the larger the dimensions $p$ is, 
the better the performance of approximation to a function is.
Also, the numerical results suggest that the choice of any orthonormal basis $\nu_i$ and its dimension $p$ does not substantially affect the performance of the estimator $\hat{r}$ by Algorithm \ref{al1}. 

\begin{table}[H]
\centering
\caption{
Estimated number $\hat{r}$ of periodicities for $100$ simulations
using B-spline basis functions, Fourier basis functions, and Wavelet basis functions, respectively ($p=1,5,10,15,20,25,30$).
}
\label{tab:result_other_basis}
\begin{tabular}{cccccccc}
\hline
       & \multicolumn{3}{c}{B-spline basis}      & \multicolumn{2}{c}{Fourier basis} & \multicolumn{2}{c}{Wavelet basis} \\ \cline{2-8} 
       & $\hat{r}=0$ & $\hat{r}=3$ & $\hat{r}=4$ & $\hat{r}=3$     & $\hat{r}=4$     & $\hat{r}=3$     & $\hat{r}=4$     \\ \hline
$p=1$  & 5           & 94          & 1           & 98              & 2               & 98              & 2               \\
$p=5$  & 0           & 97          & 3           & 97              & 3               & 98              & 2               \\
$p=10$ & 0           & 97          & 3           & 97              & 3               & 97              & 3               \\
$p=15$ & 0           & 98          & 2           & 97              & 3               & 98              & 2               \\
$p=20$ & 0           & 98          & 2           & 97              & 3               & 99              & 1               \\
$p=25$ & 0           & 98          & 2           & 98              & 2               & 100             & 0               \\
$p=30$ & 0           & 98          & 2           & 98              & 2               & 100             & 0               \\ \hline
\end{tabular}
\end{table}

\begin{figure}[h]
\centering
\includegraphics[width=\linewidth, height=0.8\textheight, keepaspectratio]{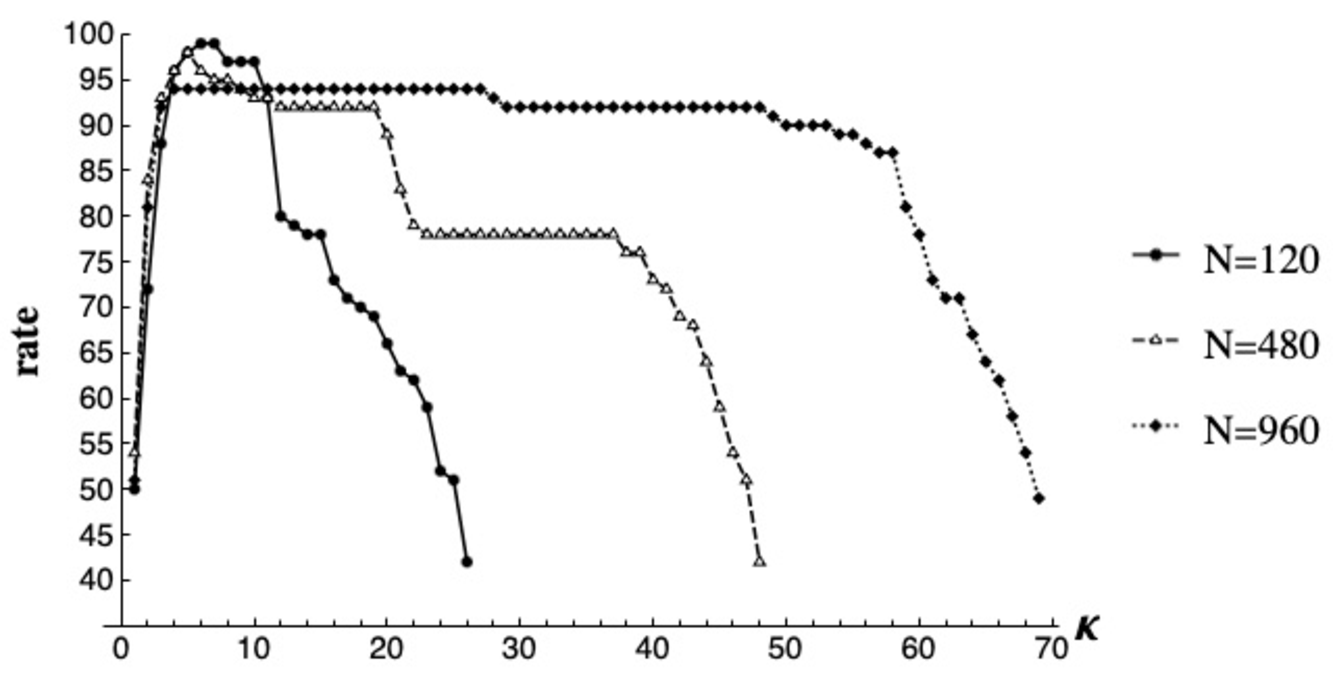}
\caption{
The rate when the periodicity is correctly estimated across all simulations for each $\kappa$ within the ``stable" range when $N=120$, $480$, and $960$, respectively.}
\label{fig:performance_hatr3}
\end{figure}

\begin{rem}
In simulations, $H$ is fixed as $H=8$.
We find that the choice of $H$ has little impact on the estimation of the number of periodicities.
In practice, the selected autoregressive order $h$ is typically small,
indicating that the results are not sensitive to the upper bound $H$.
\end{rem}

The other simulation results such as  the comparison of the ``stable "range of $\kappa$ between the true model different from \eqref{true_model} and its corresponding  local alternative model, and the comparison with AIC-type criterion,  can be found in the Supplementary Material. 
\begin{rem}
When the coefficients of other periodic components in the local alternative model are set to $20/\sqrt{N}$ or even smaller,  our method 
by Algorithm \ref{al1}
still correctly estimates the number of periodicities.
However, as a limitation of the approach,
the estimated number of true periodicities
tends to be a small value
when the coefficient parameters of functional trigonometric regression models
are small.
\end{rem}

\section{Data Analysis}
\label{sec5}
In this section, we apply our information criterion to data analyses in practice.
To be specific, 
our proposed algorithm is applied to both temperature and sunspot data to demonstrate its practical utility.
We fix the hyperparameter $\kappa$ as $\kappa = 5$,
as it lies near the lower end of the stable range and performs consistently well across all sample sizes in simulation.
Additionally, we set the maximum order of $h$ as $H=8$ and  $r\leq 10$.

\subsection{Sunspot data}
First,  we apply the proposed criterion to identify periodicities within the sunspot dataset spanning $140$ years.
The $140$-year  sunspot data comprises daily data for the entire sun from January 1, 1876, to December 31, 2015.
Missing values are handled by imputing them by the preceding data point.
We suppose $m$ consecutive data consists of a functional observation. 
The values of $m$ are specified as
$m=15$, $30$, $91$, $182$, i.e.,
corresponding to the time spans of
half month, $1$ month, $3$ months, and $6$ months, respectively.
In view of divisibility, $364$ days are supposed to be one year
when $m=91, 182$; accordingly, we have $561$ and $280$ observations 
of functional time series for each $m$.
Similarly, $360$ days are regarded as one year
when $m=15$, $30$; accordingly,
we have $3360$ and $1680$ observations of functional time series
for each $m$.

The estimated periodicity of the sunspot data is reported in Table \ref{sunspotresult} by applying Algorithm \ref{al1}.

\begin{table}[H]
\centering
\caption{
Estimates of periodicities based on the sunspot data of 140 years.
The numbers of the observations for the functional time series are $3360$, $1680$, $561$, $280$ 
for $m = 15, 30, 91, 182$, respectively.
Estimated frequencies $\hat{\theta}_1$ for different $m$ are also shown in parentheses.
}
\label{sunspotresult}
\begin{tabular}{lclcl}
\hline
      & \multicolumn{4}{c}{1st periodicity}                        \\ \cline{2-5} 
      & \multicolumn{2}{c}{[year-cycle]} & \multicolumn{2}{c}{$\hat{\theta}_1$} \\ \hline
$m=15$  & \multicolumn{2}{c}{11}         & \multicolumn{2}{c}{$\left(0.0243 \right)$}     \\
$m=30$  & \multicolumn{2}{c}{11}         & \multicolumn{2}{c}{$\left(0.0486 \right)$ }     \\
$m=91$  & \multicolumn{2}{c}{11}         & \multicolumn{2}{c}{$\left(0.146 \right)$ }     \\
$m=182$ & \multicolumn{2}{c}{11}         & \multicolumn{2}{c}{$\left(0.292 \right)$ }     \\ \hline
\end{tabular}
\end{table}

We convert the estimated frequency into the periodicity of time units. 
Each estimated frequency corresponds to approximately $11$-year-cycle.
It is well-documented that the sunspots exhibit an approximately $11$-year-cycle (e.g.~\cite{Schwabe1844}).
Table \ref{sunspotresult} demonstrates that the proposed criterion and algorithm have effectively identified the periodicity of sunspots.
Nevertheless, when using the datasets with the shorter period such as $90$ years or $60$ years, adjusting $m =15$, $30$, $91$, $182$ results in the absence of the detection of the periodicities in most cases.
The cause is attributed to the variations in the pattern of each functional time series (cf. Fig. \ref{fig_fts}).
Consequently, it is understood that longer period of the data makes periodicity detection more feasible, particularly when there are differences in the patterns of  functional time series.

\subsection{Temperature data}
Let us move to the second data analysis. 
By using the proposed criterion, we investigate the possibility of detecting periodicities in daily average temperature data from three countries along the Pacific Rim.
Along the Pacific Rim, it is known that the El Ni\~{n}o and La Ni\~{n}a phenomena occur 
with a cycle of $2$ to $7$ years, causing notable influences on temperature patterns.
Hence, we seek to verify whether it is feasible to identify the cycles of the El Ni\~{n}o and La Ni\~{n}a phenomena from daily temperature data in the countries along the Pacific Rim.

In this data analysis, we utilize the daily average temperature spanning a period of $30$ years from January 1,1990, to December 31, 2020, obtained from three countries along the Pacific Rim: Kyoto, Japan; Sydney, Australia; and Vancouver, Canada.
The average temperatures from three countries
are divided into small segments of $m= 30$, $273$ days,
corresponding to approximately $1$ month and $9$ months, respectively.

The segment length $m$ determines the time scale of the resulting functional time series. 
This may affect the set of detectable periodicities.
In this analysis, we consider two representative choices, $m=30$ and $m=273$, corresponding to different time scales.
The choice of $m=30$ corresponds to approximately one month and is used to capture seasonal variation, 
particularly the annual cycle. 
In contrast, $m=273$ represents a relatively long segment length which is not aligned with the annual cycle.

To elaborate, let us focus on the case $m=30$.
In other words, we divide the daily average temperature data for 30 years into approximately monthly intervals.
For divisibility, one year is preprocessed  into $360$ days.
This preprocessing allows us to obtain $360$ observations of the functional time series for each functional data set from the three countries.
When applying Algorithm $\ref{al1}$ to the observed functional time series, the estimated periodicities are shown in Table $\ref{table2}$.

\begin{table}[H]
\centering
\caption{
Estimates of periodicities for daily temperature data of a function with $m = 30$, i.e., 
360 observations from functional time series per country.
Estimated frequencies $\hat{\theta}_1$, $\hat{\theta}_2$, and $\hat{\theta}_3$ are also  shown in parentheses.
}
\label{table2}
\begin{tabular}{lcccc}
\hline
                                 & \multicolumn{1}{l}{} & \multicolumn{1}{l}{Japan} & \multicolumn{1}{l}{Australia} & \multicolumn{1}{l}{Canada} \\ \hline
\multirow{2}{*}{1st periodicity} & {[}year-cycle{]}     & 1                         & 1                             & 1                          \\
                                 & $\hat{\theta}_1$                    & (0.524)                         & (0.524)                              & (0.524)                           \\
\multirow{2}{*}{2nd periodicity} & {[}month-cycle{]}    & 6                         & 6                             & 6                          \\
                                 & $\hat{\theta}_2$                    & (1.05)                         & (1.05)                             & (5.24)                          \\
\multirow{2}{*}{3rd periodicity} & {[}month-cycle{]}    & 4                         & -                             & -                          \\
                                 & $\hat{\theta}_3$                    & (1.57)                         & -                             & -                          \\ \hline
\end{tabular}
\end{table}

We convert the estimated frequencies into the cycle of time units.
The 1 year-cycle, identified as the most significant periodicity among the temperature data from the three countries,
represents yearly variation.
The second most common periodicity across the three countries is a 6-month cycle, reflecting half-year variation.
The third estimated frequency, observed only in Japan, corresponds to a 4-month cycle, reflecting Japan's specific climate characteristics.

Next, we consider the case $m=273$ with 
a length of 40 observed functional time series
The results 
are shown in Table $\ref{table3}$.

\begin{table}[H]
\centering
\caption{Estimated periodicity for daily temperature data of 40 observed functions of $m=273$  per country.
Estimated frequency $\hat{\theta}_1$ is shown in parentheses.
}
\label{table3}
\begin{tabular}{rcccc}
\hline
\multicolumn{1}{l}{}             &                  & Japan & Australia & Canada \\ \hline
\multirow{2}{*}{1st periodicity} & {[}year-cycle{]} & 3     & 3         & 3      \\
                                 &       $\hat{\theta}_1$           & (1.57)     & (1.57)         & (1.57)      \\ \hline
\end{tabular}
\end{table}

We also convert this estimated periodicity into the cycle of time units. 
A $3$-year-cycle is obtained as the first estimated periodicity, revealing the El Ni\~{n}o and La Ni\~{n}a phenomena.
It is well-known
that the temperatures of the countries along the Pacific Rim 
are influenced by these phenomena.

Consequently, the proposed Algorithm \ref{al1} 
has successfully detected cycles in both sunspot activities 
and the El Ni\~{n}o and La Ni\~{n}a phenomena.
The above results are consistent with findings reported in the existing literature.
The complete results are provided in the Supplementary Material.

\section{Conclusion}
\label{sec:6}
We have proposed a new information criterion \eqref{criterion} for detecting the number of the periodicities for functional time series.
Algorithm \ref{al1} based on the new information criterion
allows for the detection of cycles by utilizing the first principal component of 
multivariate time series.
We theoretically established the consistency of
the estimates for coefficients and frequency parameters in a functional trigonometric model.
In addition, the estimated number of periodicities based on Algorithm \ref{al1} is also consistent.
Simulation studies demonstrate that  the selected model based on the new criterion is insensitive to the penalty parameter $\kappa$.
In data analyses, we identified an $11$-year-cycle in the sunspot data; 
and detected a 3-year-cycle in the daily average temperature data,
which corresponds to the well-known El Ni\~{n}o and La Ni\~{n}a phenomena.

\section*{Acknowledgements}
R.~Sagawa was supported by JST SPRING, Grant Number JPMJSP2128, Waseda Research Institute for Science and Engineering, Grant-in-Aid for Young Scientists (Early Bird), and  JEES-Mitsubishi Corporation Science and Technology Scholarship for Students, Scholarship Number  MITSUSCI2508.
Y.~Liu was supported by JSPS Grant-in-Aid for Scientific Research (C) 23K11018.
V.~Patilea acknowledges the support of the French Agence Nationale de la Recherche (ANR) 
under reference ANR-24-CE40-2439 (FUNMathStat  project).
We also appreciated the FY2024 Grant Program for Promotion of International Joint Research from Waseda University.

\section*{Supplementary Material}
\label{SM}
The proofs of technical results have been reported in the Supplementary  Material.
Also, additional simulation studies, such as 
the comparison of the ``stable" range of $\kappa$ between the true model and its corresponding local alternative model, and the comparison with AIC-type criterion, are provided.
The complete results of the data analyses can also be found in the Supplementary  Material.



\newpage

\appendix
\section{Proof of Theorem \ref{criterion_thm}}
\label{Appendix_sec5}

\begin{proof}
We assume that the frequencies $\theta_k$, $k=1,\ldots,r_{0}$, 
are ordered so that $\rho_{1} > \rho_{2} >  \cdots > \rho_{r_{0}}$,
where $\rho_{k}=\sqrt{\tilde{\alpha}_{k}^2+\tilde{\beta}_{k}^2}$,
$\tilde{\alpha}_{k}$ and $\tilde{\beta}_{k}$ are the components in the vector $\tilde{\bpsi}$ of \eqref{linear_transformation_model}.
The strict ordering is imposed only for simplifying notation.
We remark that, the least squares estimates $\hat{\tilde{\alpha}}_{r}$ and $\hat{\tilde{\beta}}_{r}$,
and the frequency estimates $\hat{\theta}_{r}$ as in Section \ref{sec2},  satisfy
\begin{align*}
 \hat{\tilde{\alpha}}_{r} &= \tilde{\alpha}_{r} + O_{p}\left(\log \log N/N\right)^{1/2}, \\
 \hat{\tilde{\beta}}_{r} &= \tilde{\beta}_{r} + O_{p}\left(\log \log N/N\right)^{1/2}, \\
 \hat{\theta}_r &= \theta_{r} + O_{p}((\log \log N)/N^{3})^{1/2},
\end{align*}
in view of proofs of Lemmas \ref{para_as_converge} and \ref{theta_estimate}.
For the sake of clarity, we define $\hat{Y}_t(r)$ as 
\[
\hat{Y}_t(r)=
\begin{cases}
\sum_{k=r+1}^{r_{0}}\{\tilde{\alpha}_{k}\cos (t\theta_k)+\tilde{\beta}_{k}\sin (t\theta_k)\}+\tilde{X}_t&r< r_{0},\\
\tilde{X}_t&r \geq r_{0}.
\end{cases}
\]
In fact, the residual $\hat{X}_t(r)$, in view of the proof of Lemma \ref{para_as_converge},  is
\begin{align*}
\hat{X}_t(r) &= \tilde{Y}_t-\hat{\bm{q}}_t(r)^{\T}\hat{\tilde{\bpsi}}(r) \\
&=
\hat{Y}_t(r)+ O_{p}\left(\frac{\log \log N}{N}\right)^{1/2}.
\end{align*}
This leads to the approximation 
\[
\frac{1}{N}\sum_{t=1}^N\hat{X}_t(r)\hat{X}_{t-l}(r)=\frac{1}{N}\sum_{t=1}^N\hat{Y}_t(r)\hat{Y}_{t-l}(r)+ O_{p}\left(\frac{\log \log N}{N}\right).
\]
To compare the approximate errors $\hat{\sigma}_{r-1}^2(h)$ and $\hat{\sigma}_{r}^2(h)$, 
applying Lemma $\ref{error_evaluate}$ yields
\begin{align*}
\hat{\sigma}_{r-1}^2(h)&=\hat{\sigma}^2(h)+\sum_{k=r}^{r_{0}}\frac{|\eta_h(e^{\mi \theta_k})|^24\pi f_{\tilde{X}}(\theta_k)}{h}+o(h^{-1}),\\
\hat{\sigma}_{r}^2(h)&=\hat{\sigma}^2(h)+\sum_{k=r+1}^{r_{0}}\frac{|\eta_h(e^{\mi \theta_k})|^24\pi f_{\tilde{X}}(\theta_k)}{h}+o(h^{-1}).\\
\end{align*}
Thus, we have
\[
\hat{\sigma}_{r-1}^2(h)-\hat{\sigma}_{r}^2(h)=\frac{|\eta_h(e^{\mi \theta_r})|^24\pi f_{\tilde{X}}(\theta_r)}{h}+o(h^{-1}),
\]
and equivalently,
\begin{align*}
\frac{\hat{\sigma}_{r-1}^2(h)}{\hat{\sigma}_{r}^2(h)}=1+\frac{|\eta_h(e^{\mi \theta_r})|^24\pi f_{\tilde{X}}(\theta_r)}{h\hat{\sigma}_{r}^2(h)}+o(h^{-1}).
\end{align*}
Taking the logarithm on both sides, we obtain
\begin{align*}
\log \frac{\hat{\sigma}_{r-1}^2(h)}{\hat{\sigma}_{r}^2(h)}&=\log \left\{1+\frac{|\eta_h(e^{\mi \theta_r})|^24\pi f_{\tilde{X}}(\theta_r)}{h\hat{\sigma}_{r}^2(h)}+o(h^{-1})\right\}\\
&=\frac{|\eta_h(e^{\mi \theta_r})|^24\pi f_{\tilde{X}}(\theta_r)}{h\hat{\sigma}_{r}^2(h)}+o(h^{-1}),
\end{align*}
and hence,
\[
\log \{\hat{\sigma}_{r-1}^2(h)\}=\log \{\hat{\sigma}_{r}^2(h)\}+\frac{|\eta_h(e^{\mi \theta_r})|^24\pi f_{\tilde{X}}(\theta_r)}{h\hat{\sigma}_{r}^2(h)}+o(h^{-1}).
\]
For the above discussion, 
\begin{align*}
\varphi(r-1,h)&=\log \{\hat{\sigma}_{r-1}^2(h)\}+\{\kappa(r-1)+h\}\frac{\log N}{N}\\
&=\log \{\hat{\sigma}_{r}^2(h)\}+\frac{|\eta_h(e^{\mi \theta_r})|^24\pi f_{\tilde{X}}(\theta_r)}{h\hat{\sigma}_r^2(h)}+\{\kappa(r-1)+h\}\frac{\log N}{N}+o(h^{-1})\\
&=\log \{\hat{\sigma}_{r}^2(h)\}+\{\kappa r+h\}\frac{\log N}{N}+\left\{\frac{|\eta_h(e^{\mi \theta_r})|^24\pi f_{\tilde{X}}(\theta_r)}{h\hat{\sigma}_r^2(h)}-\kappa\frac{\log N}{N}\right\}+o(h^{-1})\\
&>\log \{\hat{\sigma}_{r}^2(h)\}+\{\kappa r+h\}\frac{\log N}{N}\\
&=\varphi(r,h).
\end{align*}
Consequently, for $r\leq r_{0}$, we obtain
\[
\varphi(r-1,\hat{h}_{r-1})>\varphi(r,\hat{h}_r).
\]

Next, we consider $r> r_{0}$.
In view of the proof of Lemmas \ref{para_as_converge} and \ref{theta_estimate} again, we have
\begin{align*}
 \hat{\tilde{\alpha}}_{r} &= \tilde{\alpha}_{r} + o_{p}\left(\log \log N/N\right)^{1/2}, \\
 \hat{\tilde{\beta}}_{r} &= \tilde{\beta}_{r} + o_{p}\left(\log \log N/N\right)^{1/2}, \\
 \hat{\theta}_r &= \theta_{r} + o_{p}((\log \log N)/N^{3})^{1/2}.
\end{align*}
Similarly, noting that $\hat{Y}_t(r) = \tilde{X}_t$, we have
\[
\frac{1}{N}\sum_{t=1}^N\hat{X}_t(r)\hat{X}_{t-l}(r)=\frac{1}{N}\sum_{t=1}^N\tilde{X}_t \tilde{X}_{t-l} +o_{p}\left(\frac{\log \log N}{N}\right),
\]
with the difference as
\[
\hat{\sigma}^2_r(h)=\hat{\sigma}^2_{r-1}(h)+o_{p}\left(h\frac{\log \log N}{N}\right).
\]
Therefore, we obtain
\begin{align*}
\varphi(r,h)&=\log \{\hat{\sigma}^2_r(h)\}+(\kappa r+h)\frac{\log N}{N}\\
&=\log \{\hat{\sigma}^2_{r-1}(h)\}+(\kappa r+h)\frac{\log N}{N}+o_{p}\left(h\frac{\log \log N}{N}\right)\\
&>\log \{\hat{\sigma}^2_{r-1}(h)\}+(\kappa r+h)\frac{\log N}{N}-\kappa\frac{\log N}{N}+o_{p}\left(h\frac{\log \log N}{N}\right)\\
&=\varphi(r-1,h)
\end{align*}
and $\varphi(r,h)$ is an increasing function of $r$ for $r>r_{0}$.
Thus, for $r>r_{0}$, we obtain
\[
\varphi(r,\hat{h}_r)>\varphi(r-1,\hat{h}_{r-1}).
\]

In summary,
\[
\begin{cases}
\varphi(r-1,\hat{h}_{r-1})>\varphi(r,\hat{h}_r)&r\leq r_{0},\\
\varphi(r,\hat{h}_r)>\varphi(r-1,\hat{h}_{r-1})&r > r_{0}.
\end{cases}
\]
Consequently, 
$\hat{r}$, minimizing $\varphi(r,\hat{h}_r)$, is consistent with the true order $r_{0}$.  
\end{proof}


\end{document}